\documentclass[conference]{IEEEtran}
\IEEEoverridecommandlockouts
\usepackage{cite}
\usepackage{amsmath,amssymb,amsfonts}
\usepackage{algorithmic}
\usepackage{graphicx}
\graphicspath{{Images//}} 

\usepackage{textcomp}
\usepackage{xcolor}
\def\BibTeX{{\rm B\kern-.05em{\sc i\kern-.025em b}\kern-.08em
    T\kern-.1667em\lower.7ex\hbox{E}\kern-.125emX}}

\usepackage{psfrag}
\usepackage{amsthm}
\usepackage{subfigure}

\label{newcommand}
\label{English Chars}
\newcommand{\bA}{\textbf{A}}
\newcommand{\ba}{\textbf{a}}

\newcommand{\bc}{\textbf{c}}

\newcommand{\bP}{\textbf{P}}

\label{Roman Chars}

\label{Theorems}
\newtheorem{theorem}{Theorem}%[section]

\newtheorem{lemma}{Lemma}

\theoremstyle{remark}

\newtheorem{remark}{Remark}

\begin{document}

\title{Decentralizing Multi-Operator Cognitive Radio Resource Allocation: An Asymptotic Analysis}

\author{%
	\IEEEauthorblockN{Ehsan Tohidi\IEEEauthorrefmark{1},
		David Gesbert\IEEEauthorrefmark{1},
		Antonio Bazco-Nogueras\IEEEauthorrefmark{1}\IEEEauthorrefmark{2},
		Paul de Kerret\IEEEauthorrefmark{1}}
	\IEEEauthorblockA{\IEEEauthorrefmark{1}%
		Communication Systems Department, EURECOM,
		06410 Biot, France,
		\{tohidi,gesbert,bazco,dekerret\}@eurecom.fr}
	\IEEEauthorblockA{\IEEEauthorrefmark{2}%
		Mitsubishi Electric R\&D Centre Europe (MERCE),
		35700 Rennes, France.}
		\thanks{Paul de Kerret is currently working at Mantu R\&D lab, Biot, France.}
}

\maketitle

\begin{abstract}
We address the problem of resource allocation (RA) for spectrum underlay in a cognitive radio (CR) communication
system with multiple secondary operators sharing resource with an incumbent primary operator. The multiple secondary operator RA problem is well known to be 
especially challenging because of the inter-operator coupling constraints arising in the optimization problem, which render impractical inter-operator information exchange necessary. In this paper, we consider a satellite setting for multi-operator CR. In the CR maturation regime, i.e., the period in which the secondary subscriber density is growing yet remains much below that of incumbent users, we show that in fact the inter-operator mutual constraints can be neglected, thus making distributed (across secondary operators) optimization possible. Furthermore, we establish analytically that the mutual constraints asymptotically vanish with the primary user density.
\end{abstract}

\begin{IEEEkeywords}
Cognitive Satellite Communication, Resource Allocation, Asymptotic Limit, Distributed Optimization
\end{IEEEkeywords}

\section{Introduction}
Explosive growth in both quantity and quality of services, on one hand, and scarcity of spectrum, on the other hand, has directed the research directions toward enhancing the spectrum utilization. In this context, CR communication systems have been introduced as a promising solution \cite{1391031,788210,7748543}.
%\cite{7748543}. 
In CR systems, there exists an incumbent network, i.e., the primary system, which has the license of using the spectrum. There is also one or more secondary systems that aim to utilize the same spectrum and maximize the data rate of the secondary users (SUs), however with a guarantee of not excessively interfering with the primary system. This guarantee is maintained through some interference temperature thresholds which are held for each primary user (PU). Classically, an optimization problem with the SUs sum-rate as the objective function and a set of constraints on interference imposed on PUs can be formulated. Although the optimization problem (subband assignment, power allocation, etc.) is mixed-integer and usually NP-hard \cite{5963799}, it is yet well-defined and several algorithms have been proposed to find a proper solution \cite{8546798,8500158,7336495,7510839}. 

The problem of RA for CR systems has been studied vastly as briefly reviewed next. For instance, considering a single operator scenario, a message-passing algorithm is proposed in \cite{8546798} to perform joint subband assignment and power allocation, while in \cite{8500158}, a deep learning-based RA algorithm for heterogeneous internet of things (IoT) is presented. Also, in \cite{7336495,7510839}, joint power and carrier RA algorithms for cognitive satellite (CogSat) communications with incumbent terrestrial networks are proposed.

The RA problem becomes more challenging when several secondary operators are co-existing in an underlay manner with a common incumbent, referred below as a \textit{multi-operator CR system}. In this case, SUs from different secondary operators contribute to the interference level at each PU and must coordinate to keep the total interference low. Consequently, in order to satisfy the interference constraints, secondary operators must typically share their user resource information with each other towards finding a centralized solution \cite{5955145,7516556,7039232}.
For instance, in \cite{5955145}, a centralized dynamic spectrum allocation scheme is proposed which can measure the interference level and interact dynamically to minimize interference and enhance spectrum utilization while maintaining a satisfactory level of quality of service (QoS). Similarly, \cite{7516556,7039232} proposed approaches of spectrum sharing that involve an exchange of information among operators in order to increase spectrum utilization. 
In most scenarios, information sharing among otherwise competing operators raises complexity and privacy issues. Hence, a form of inter-operator coordination that can circumvent this problem while meeting maximum allowed interference constraints is highly desirable.
In this paper, we investigate this problem in the particular context of CogSat communications, where the secondary operators serve their subscribers via satellites, while the incumbent is a terrestrial cellular or satellite operator \cite{7060478}. A regime of particular interest is the so-called CR maturation regime in which the number of SUs is growing yet remains much below the number of PUs.

In this paper, we consider the CR maturation regime and analyze the coupling effects between secondary operators towards the PU received interference. We establish analytically that such mutual constraints asymptotically
vanish with the primary user density. As a result,
we obtain that the inter-operator mutual constraints can be neglected, thus making distributed (i.e., across secondary
operators) optimization possible. In turn, this result indicates that the secondary operators can in fact coordinate without the need for private information exchange.

\section{System Model}
Consider the problem of RA in a multi-operator CogSat communication system. We assume $L$ PUs forming the primary network and licensed to communicate over the spectrum. Moreover, $N$ secondary operators are considered which are serving $K$ SUs and forming the secondary network. Each secondary operator has a satellite with $B$ beams sharing the same bandwidth. This bandwidth is split up into $M$ subbands. In each beam of each satellite, we assume a frequency division multiple access (FDMA) serving up to $M$ SUs, i.e., one subband per SU. We denote the number of SUs of each secondary operator with $Q=BM$ (therefore $K=NQ$).

An example of CogSat communication system is depicted in Fig. \ref{fig:Scenario}, where $N=2$ operators each with $B=2$ beams are communicating with $M=3$ SUs per beam. Therefore, $Q=6$ and $K=12$ SUs per operator and in total, respectively. Also, there are $L=5$ PUs communicating through the incumbent network. In this paper, we consider the uplink channel for the secondary network.
\begin{figure}
	\centering
	\includegraphics[width=.48\textwidth]{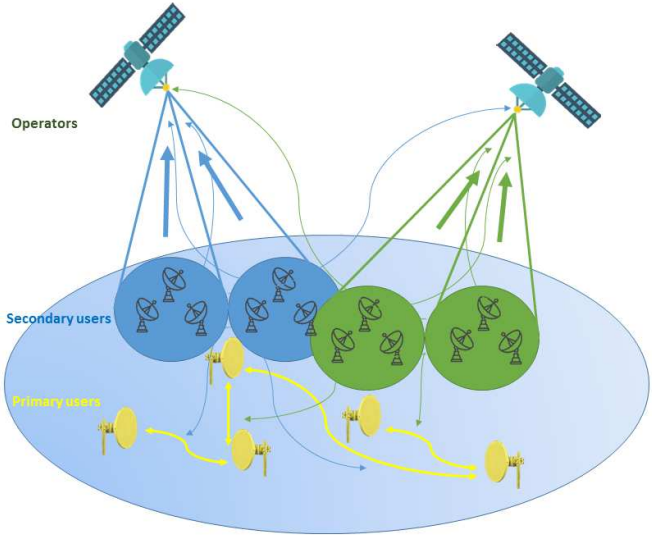}			
	\caption{A multi-operator CogSat communication system.}
	\label{fig:Scenario}
\end{figure}

We define $G_{n,q}^{(b)}(m)$ as the channel gain in subband $m$ from SU $q$ of satellite $n$ to beam $b$. Moreover, $F_{n,q}^{(l)}(m)$ is the channel gain on subband $m$ from SU $q$ of operator $n$ to PU $l$. We model the channel gains among SUs and PUs in the following form
\begin{equation}
F_{n,q}^{(l)}(m) = \frac{C}{r_{n,q,l}^{\alpha}} + s_{n,q}^{(l)}(m),
\label{channelgainshadow}
\end{equation}
where $r_{n,q,l}$ is the distance between SU $q$ of operator $n$ and PU $l$, $\alpha$ is the path loss exponent, e.g., $\alpha=2$ for free space, $C$ is a constant, and $s_{n,q}^{(l)}(m)$ describes the shadowing effect which are assumed to be i.i.d. and characterized by $s_{n,q}^{(l)}(m)\sim \mathcal{N}(0, \sigma_s^2)$. The set of SUs of beam $b$ of operator $n$, operator $n$, and the total set of SUs are denoted by $\mathcal{U}_{n,b}$, $\mathcal{U}_n$, and $\mathcal{U}$, respectively.
In Table \ref{tableparam}, we recall the list of parameters.
\begin{table}
%\resizebox{\columnwidth}{!}{	
	\begin{center}
	\caption{List of parameters}
	\label{tableparam}
	\begin{tabular}{| c | c |} 
		\hline
		Parameter & Description \\
		\hline\hline
		$N, B$ & the number of operators and beams, respectively \\ 
		\hline
		$M$ & the number of frequency subbands and the number\\& of SUs in each beam \\
		\hline
		$L$ & the number of PUs \\
		\hline
		$Q,K$ & the number of SUs of each operator and the total \\&number of SUs, respectively \\
		\hline
		$\mathcal{U}_{n,b}, \mathcal{U}_n,\mathcal{U}$ & the set of SUs of beam $b$ of operator $n$, operator $n$,\\& and in total, respectively \\
		\hline
		$F_{n,q}^{(l)}(m)$ & the channel gain on subband $m$ from SU $q$ \\&of operator $n$ to PU $l$ \\
		\hline
		$G_{n,q}^{(b)}(m)$ & the channel gain on subband $m$ from SU $q$ \\& of operator $n$ to beam $b$ \\
		\hline
	\end{tabular}
\end{center}
%}
\end{table}

\section{Problem Formulation}
Considering the uplink channel, the goal is to maximize the sum-rate of SUs while not excessively interfering with the PUs. The optimization variables are subband assignment and power allocation which are defined as follows:
\begin{itemize}
	\item $\bA = [\bA(1),...,\bA(m)]\in\{0,1\}^{N\times Q\times M}$ the subband assignment matrix where $A_{n,q}(m)$ is 1 if subband $m$ is assigned to the SU $q$ of operator $n$ and $0$ otherwise.
	\item $\bP = [\bP(1),...,\bP(m)]\in \mathbb{R}_+^{N\times Q\times M}$ the power allocation matrix where $P_{n,q}(m)$ corresponds to the transmission power of SU $q$ of operator $n$ on subband $m$.
\end{itemize}
Therefore, the total interference imposed by the SUs on PU $l$ in subband $m$ is calculated in the following form
\begin{equation}
I^{(l)}(m)=\sum_{n=1}^N\sum_{q=1}^Q{{A_{n,q}(m)F_{n,q}^{(l)}(m)P_{n,q}(m)}}.
\label{interferPU}
\end{equation}

To guarantee the communication quality of the primary network, we consider average interference-temperature constraints associated with the $L$ PUs and for each of the $M$ subbands. The constraints are represented as follows
\begin{equation}
\mathbb{E}\{I^{(l)}(m)\} \leqslant I^{(l)}_{th}(m), \forall l,m,
\label{expectedinterference}
\end{equation}
where $I_{th}^{(l)}(m)$ is the interference-temperature threshold at PU $l$ on subband $m$.

As the interference imposed on PUs in \eqref{interferPU} is a linear function of channel gains, the shadowing parameter in \eqref{channelgainshadow} diminishes due to the averaging operator in \eqref{expectedinterference}. Therefore, to simplify the notations in derivations of Section IV, we only consider the deterministic part, i.e., as a function of distance.

The signal power for SU $q$ of secondary operator $n$ in beam $b$, i.e., $q\in\mathcal{U}_{n,b}$, on subband $m$ at the satellite is $A_{n,q}(m)G_{n,q}^{(b)}(m)P_{n,q}(m)$, while all transmissions by the other SUs at the satellite $n$ in beam $b$ play the role of interference for this user. Thus, for SU $q\in\mathcal{U}_{n,b}$, the received interference on subband $m$ for beam $b$ of satellite $n$ is given by
\begin{equation}
J^{(b)}_{n,q}(m) = \sum_{\substack{i=1 \\ i\neq q}}^Q{A_{n,i}(m) G_{n,i}^{(b)} P_{n,i}(m)},
\end{equation}
where inter-operator interference is not considered since satellites are assumed to be far from each other and the SUs have highly directed radiation patterns toward their associated satellite \cite{maral2011satellite,7336495}. Furthermore, due to the FDMA scheme, there is no intra-beam interference among SUs.

We consider the sum-rate of SUs as the objective function.
The sum-rate for the secondary operator $n$ is calculated as follows
\begin{equation}
\begin{aligned}
R_n =& \sum_{b=1}^B\sum_{q\in\mathcal{U}_{n,b}}{\sum_{m=1}^M{A_{n,q}(m)}}\\& \ \  \times \log_2(1+\frac{G_{n,q}^{(b)}(m)P_{n,q}(m)}{1+J^{(b)}_{n,q}(m)}).
\end{aligned}
\end{equation}
Subsequently, the total sum-rate is given by
\begin{equation}
R = \sum_{n=1}^NR_n.
\end{equation}

Considering a peak power constraint on each subband for each SU, i.e., $0 \leqslant P_{n,q}(m) \leqslant P_{\max}, \forall n,q$, the RA optimization problem can be formulated in the following form
\begin{equation}
\begin{aligned}
\max_{\bA,\bP}\quad & R \\
\text{s.t.}\quad& \text{C1: } \sum_{n=1}^N\sum_{q=1}^QA_{n,q}(m)F_{n,q}^{(l)}(m)P_{n,q}(m) %\leqslant I_{th}^{(l)}(m), \forall l,m, 
\\& \quad \quad \quad \quad \quad \quad \quad \quad \quad \quad \leqslant I_{th}^{(l)}(m), \quad \forall l,m, \\
& \text{C2: } A_{n,q}(m) \in \{0,1\}, \quad \forall n,q,m, \\
& \text{C3: } 0 \leqslant P_{n,q}(m) \leqslant P_{\max}, \quad \forall n,q,m, \\
& \text{C4: } \sum_{q\in\mathcal{U}_{n,b}}{A_{n,q}(m)}=1, \quad \forall n,b,m, \\
& \text{C5: } \sum_{m=1}^M{A_{n,q}(m)}=1, \quad \forall n,q, \\
\end{aligned}
\label{equ:mainopt}
\end{equation}
where C1 is to ensure that the interference level does not exceed the given threshold, C2 states that the subband assignment is binary, C3 is to limit the SU power between 0 and the maximum allowed level $P_{\max}$, and subband assignment restrictions, i.e., one SU in each beam be assigned to each subband and one subband be assigned to each SU, are applied in C4 and C5. Although the optimization problem \eqref{equ:mainopt} is known to be NP-hard \cite{5963799}, the main challenge is how to solve the problem in a distributed framework. In this paper, our approach is different and we want to analyze the problem asymptotically and demonstrate how the problem simplifies in the asymptotic regime.

\section{Asymptotic Limit of CogSat Networks}
In the asymptotic scenario, we assume that the number of PUs $L$ grows faster than the number of SUs $K$ while both are asymptotically large numbers. The asymptotic scenario is mathematically modelled as $K = L^\beta$, with the scaling exponent $0<\beta<1$, and $L\to\infty$. We also define the asymptotic index $\lambda=\frac{L}{K} = L^{1-\beta}$. If the PUs and SUs are located randomly with a uniform distribution over the region of interest, we expect to encounter a scenario such as Fig. \ref{fig:AsymptoticFS}, where PUs are densely located and SUs are far from each other relatively. We mathematically verify this observation and use it to simplify the optimization problem in \eqref{equ:mainopt}, where we aim to show that under the asymptotic limit conditions, i.e., $L\to\infty$ and $\lambda\to\infty$, the optimization problem \eqref{equ:mainopt} uncouples such that the global optimal solution can be reached through a distributed optimization. More precisely, we demonstrate how the set of interference constraints in C1 converts to simple constraints on the peak power.

\begin{figure}
	\centering
	\includegraphics[width=.48\textwidth]{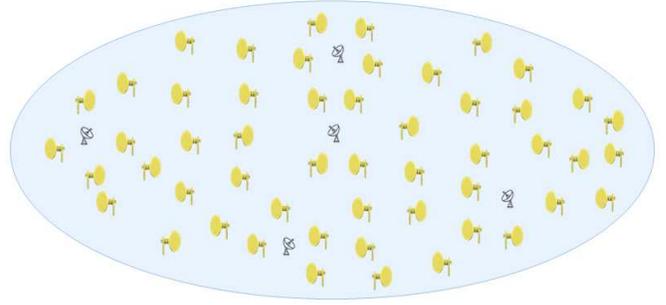}		
	\caption{A CogSat communication setting with the number of PUs growing faster than the number of SUs.}
	\label{fig:AsymptoticFS}
\end{figure}

In order to increase the readability of the analysis, we simplify the notation in this section. First of all, we drop the operator and beam indices. Then, as we study the interference constraints on PUs and the constraints are imposed for each subband frequency separately, we categorize the SUs based on their assigned subband. Following the FDMA scheme for SUs of each beam and due to frequency reuse in different beams, each subband is assigned to exactly one SU per beam. Consequently, we have $\bar{K}=NB$ SUs transmitting in each subband considering all the $N$ secondary operators. In the following, we present w.l.o.g. the analysis for one subband, although it applies to every subband.

For $\bar{k}\in\{1,...,\bar{K}\}$, we represent the channel gain and distance between SU $\bar{k}$ and PU $l$ by $F_{\bar{k},l}$ and $r_{\bar{k},l}$, respectively. Moreover, we denote the distance of the closest PU to SU $\bar{k}$ and the corresponding channel gain by $d_{\rm{SU},\bar{k}}$ and $F_{\bar{k}}$, respectively, and the distance of the closest SU to PU $l$ is denoted by $d_{\rm{PU},l}$. Also, we use the simplified notation $I_{th}$ for the interference-temperature threshold.

\begin{remark}
	\label{remarkK}
	In most scenarios, the number of operators $N$ is a small number and the number of beams is reported to be less than the number of subbands, i.e., $B\leqslant M$, \cite{7336495,louchart2019resource}, therefore we have $\bar{K}^2 \leqslant K$.
\end{remark}

\begin{remark}
	All the following lemmas are proved for a circular shaped area with a unity radius. However, they can be applied to any arbitrary shape by only considering the smallest circle that contains that area, i.e., peripheral circle, which multiplies the dimension with a constant number (i.e., not dependent on $L$ or $K$). Also, unity size is only assumed to simplify the illustration and, since normalization affects both numerator and denominator in ratios, it does not reduce the generality of the derivations.
\end{remark}

\begin{lemma}
	\label{lemma1}
	For SU $\bar{k}$, we have the following upper bound on the allocated power:
	\begin{equation}
	P_{\bar{k}} \leqslant \frac{I_{th}}{F_{\bar{k}}}.
	\label{equ:looseupperbound1}
	\end{equation}
\end{lemma}
\begin{proof}
	The proof is straightforward. Considering the interference constraint for the closest PU to SU $\bar{k}$ similar to C1 in \eqref{equ:mainopt}, the summation of interference from all SUs should not exceed $I_{th}$. Since all $P_i$s and channel gains are considered to be non-negative, the upper bound in \eqref{equ:looseupperbound1} is obtained.	
\end{proof}

\begin{lemma}
	\label{lemma3}
	For each arbitrary SU $\bar{k}$ and for a given constant $0<\epsilon <1$, we have the probabilistic upper bound for $d_{\rm{SU},\bar{k}}$ (i.e., distance to the closest PU) as
	\begin{equation}
	\lim_{L\to \infty} p(d_{\rm{SU},\bar{k}}>\frac{1}{\sqrt{L^{1-\epsilon}}})= 0.
	\label{equ:r_k}
	\end{equation}
\end{lemma}
\begin{proof}
	The proof is derived in Appendix \ref{proof lemma 3}.	
\end{proof}

\begin{lemma}
	\label{lemma4}
	For each arbitrary PU $l$ and for a given constant $0<\epsilon <1$, we have the probabilistic lower bound for $d_{\rm{PU},l}$ as
	\begin{equation}
	\lim_{\bar{K}\to \infty} p(d_{\rm{PU},l}<\frac{1}{\sqrt{\bar{K}^{1+\epsilon}}})= 0.
	\label{equ:d_l}
	\end{equation}
\end{lemma}
\begin{proof}
	The proof is derived in Appendix \ref{proof lemma 4}.	
\end{proof}

\begin{theorem}
	\label{theorem1}
	The set of interference-temperature constraints in the optimization problem \eqref{equ:mainopt}, i.e., C1, converts to the peak power constraints given that $L\to\infty$:
	\begin{equation}
	P_{\bar{k}} \leqslant \frac{I_{th}}{F_{\bar{k}}}.
	\end{equation}	
\end{theorem}
\begin{proof}
	The proof is derived in Appendix \ref{proof theorem 1}.			
\end{proof}

Combining the result of Theorem \ref{theorem1} with the original peak power constraint, i.e., C3, we achieve the following peak power constraints:
\begin{equation}
0 \leqslant P_{\bar{k}} \leqslant \min(\frac{I_{th}}{F_{\bar{k}}},P_{\max}), \quad \forall {k}.
\end{equation}

\begin{remark}
	In practice and in scenarios with $\lambda>> 1$ but $\lambda \not\to \infty$, we need to keep the factor $(1-\frac{1}{\lambda})$ for the power constraints (see \eqref{equ:constraintswithlambda} in Appendix \ref{proof theorem 1} for more details).
\end{remark}

\section{Numerical Experiments}
In this section, we evaluate the accuracy of the mathematical analysis for finite scenarios in a CogSat communication network and also demonstrate the performance improvement due to the enabled decentralized optimization. 
The PUs and SUs are randomly located with a uniform distribution in a region with a square shape.

In the first scenario, we fix  $\lambda=10$ and evaluate how the probabilities derived in lemmas \ref{lemma3} and \ref{lemma4} and Theorem \ref{theorem1} converge versus ${L}$. The parameter $\epsilon$ is set to $0.5$. Fig. \ref{Pl} depicts the probability of far PUs to a SU, $P_f = p(d_{\rm{SU},\bar{k}}> \frac{1}{\sqrt{L^{1-\epsilon}}})$. As shown in Fig. \ref{Pl}, even without growing to infinity, the simulation result approaches the theoretical asymptotic derivation.
In Fig. \ref{Psm}, the probability of close SUs to a PU, $P_{c} = p(d_{\rm{PU},l}<\frac{1}{\sqrt{\bar{K}^{1+\epsilon}}})$ is plotted. Although the probability is not close to zero in this plot, as illustrated in Fig. \ref{Psm}, it is converging to zero by increasing the number of PUs.
Finally, in Fig. \ref{Ps}, following the new peak power constraint, we observe the average probability of satisfying interference-temperature constraints of PUs, $P_s$. In this set of simulations, we do not consider the original peak power constraint ($P_{max}$) to make the effect of the new constraint more evident. It is worth noting that adding the peak power constraint causes the SUs powers to remain the same or decrease, and consequently $P_s$ will remain the same or increase. As observed in Fig. \ref{Ps}, by increasing the number of users the result is closer to the derivations obtained for the asymptotic condition.

\begin{figure*}
	\centering	
	\subfigure[] {\includegraphics[width=.32\textwidth]{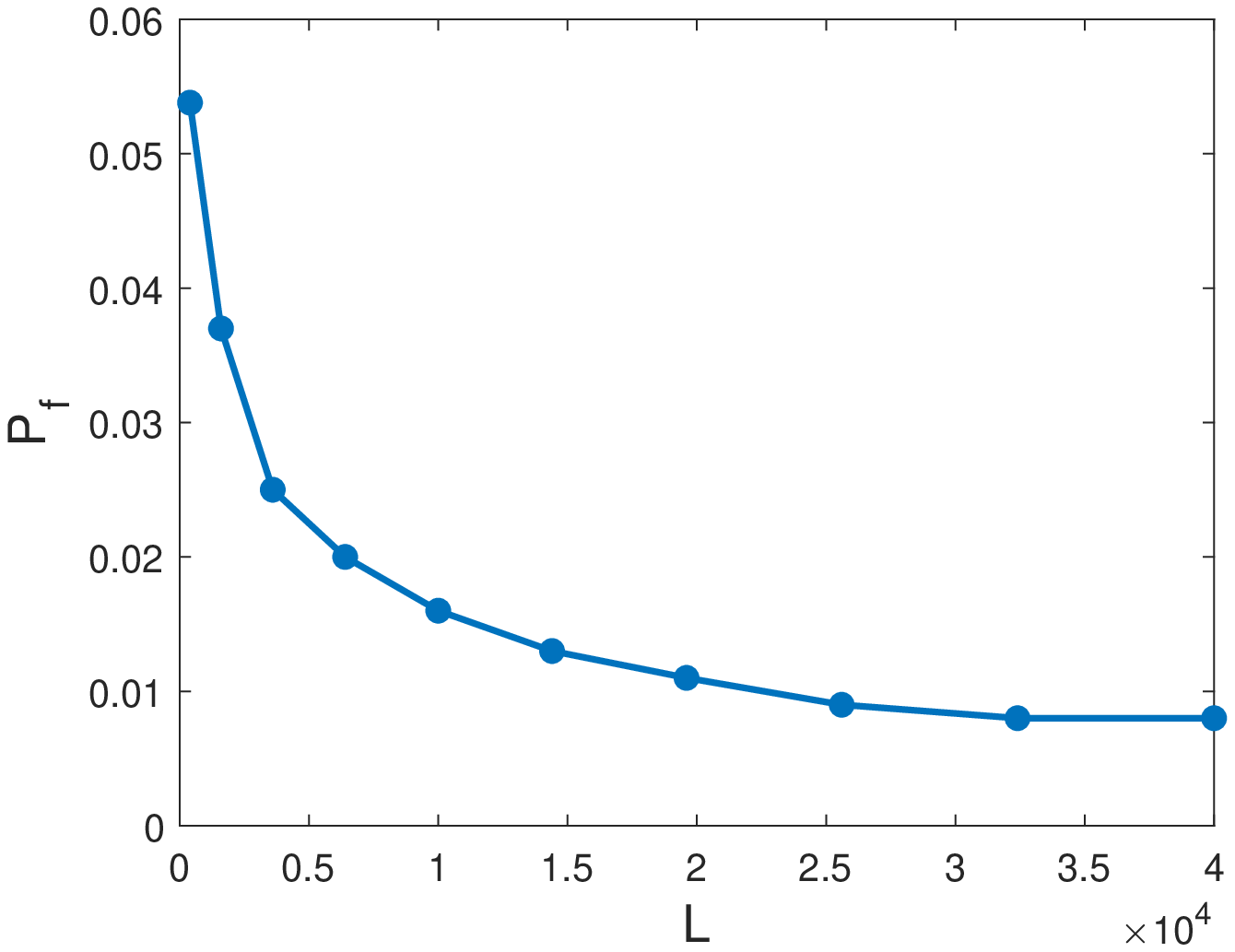}
		\label{Pl}}
	\subfigure[] {\includegraphics[width=.32\textwidth]{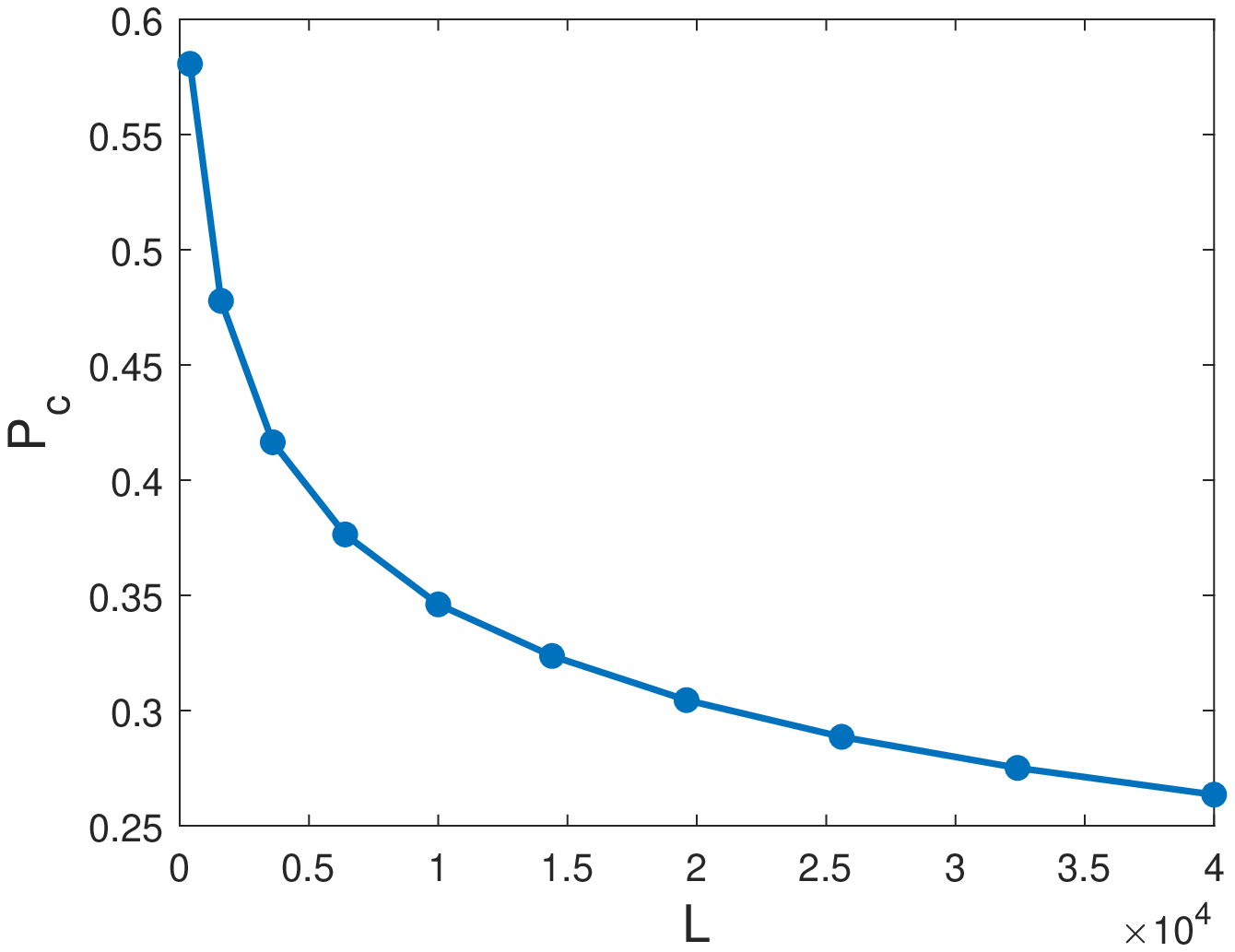}
		\label{Psm}}
	\subfigure[] {\includegraphics[width=.32\textwidth]{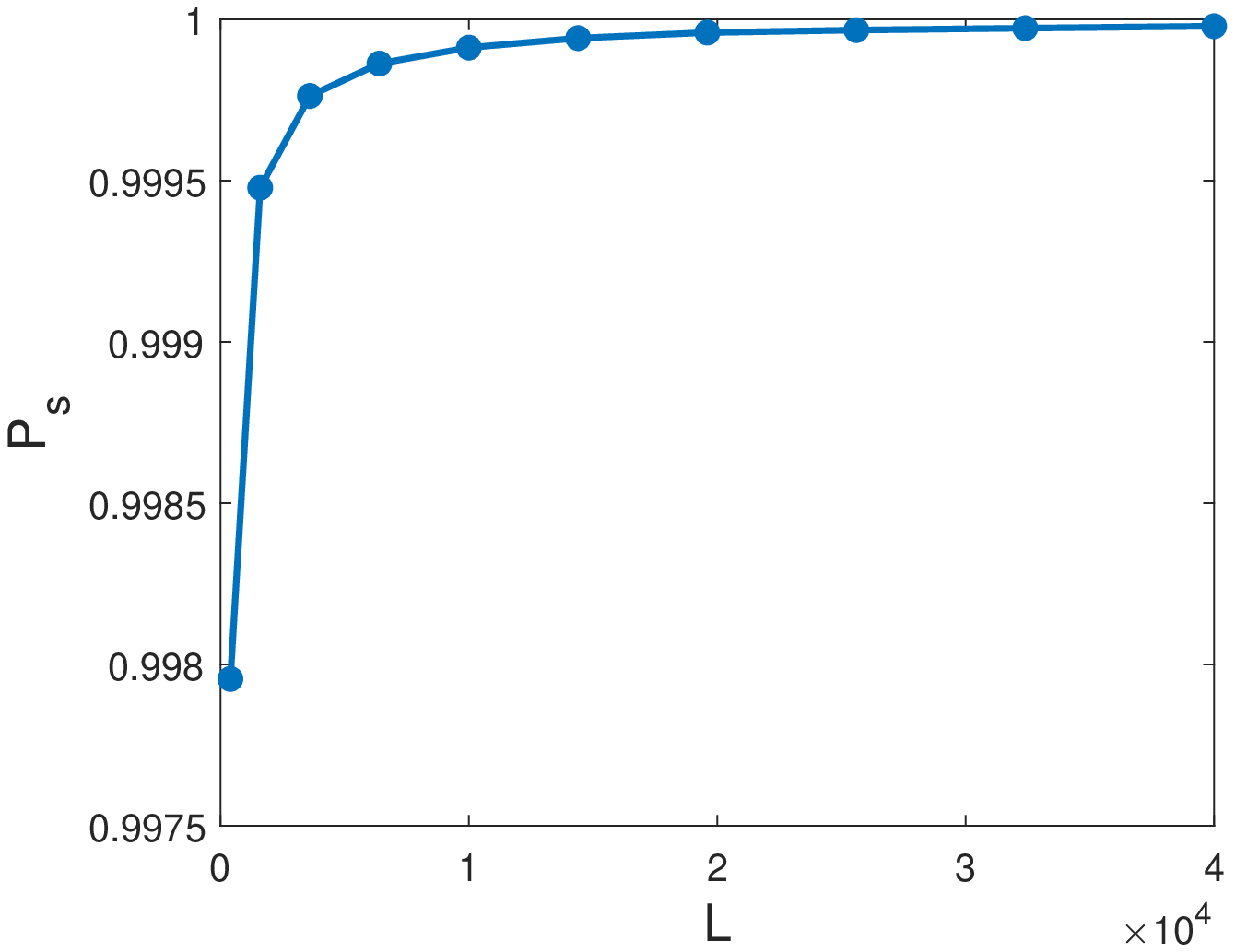}
		\label{Ps}} \\
	\caption{Evaluation of the asymptotic limit scenario versus number of SUs for $\lambda=10$, (a) probability of far PUs to a SU (Lemma \ref{lemma3}), (b) probability of close SUs to a PU (Lemma \ref{lemma4}), (c) average probability of satisfying interference-temperature constraints of PUs (Theorem \ref{theorem1}). }
	\label{simulfigs}
\end{figure*}

Next, we simulate a CogSat communication system to observe the achieved sum-rate provided by the decentralizing scheme (the optimal method). Since the decentralized scheme does not need any inter-operator information exchange, for a fair performance comparison, we consider the equal split algorithm in which the interference-temperature thresholds are equally split among the operators and each operator performs a local optimization algorithm. In this scenario, we consider $5$ satellites. To solve the optimization problems in this scenario, we employ the convex relaxation techniques in \cite{louchart2019resource} followed by a rounding algorithm \cite{8537943} to obtain a feasible solution for the original problem.
 Fig. \ref{fig:Distributed-Equal} presents the sum-rate of both methods versus the number of PUs $L$. As shown in Fig. \ref{fig:Distributed-Equal}, since the number of constraints is proportional to $L$, the sum-rates of both methods decrease by increasing $L$. The decentralized method outperforms the equal split method in terms of sum-rate. This improvement increases by increasing $L$, because the equal split method becomes more limiting as the number of constraints increases.
\begin{figure}
	\centering
	\includegraphics[width=.48\textwidth]{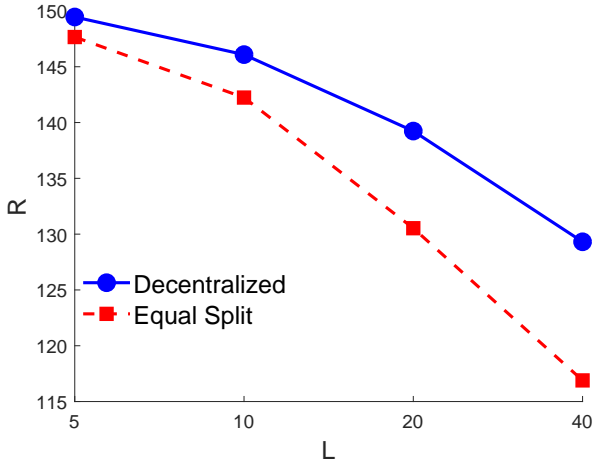}			
	\caption{Sum-rate comparison between the decentralized and equal split methods in a CogSat communication system.}
	\label{fig:Distributed-Equal}
\end{figure}

\section{Conclusion}
In this paper, we have considered the asymptotic limit of CogSat communication systems where the number of PUs grows faster than the number of SUs. 
We demonstrated that the coupling interference-temperature constraints become simple peak power constraints in the asymptotic regime. The result is surprising because an increase in the number of PUs leads to more interference-temperature constraints and subsequently the users are more coupled. This observation is especially important in multi-operator settings where solving the problem requires either information exchange among the operators or suboptimally splitting the shared resource among operators. 

\appendices

\section{Proof of lemma \ref{lemma3}}
\label{proof lemma 3}
\begin{proof}
	
	Since $d_{\rm{SU},\bar{k}}$ is the distance of the closest PU to the SU $\bar{k}$, the probability of $d_{\rm{SU},\bar{k}}>\frac{1}{\sqrt{L^{1-\epsilon}}}$ is equal to the probability of no PU being inside a circle with a radius equal to $\frac{1}{\sqrt{L^{1-\epsilon}}}$ centered at SU $\bar{k}$. Therefore,
	\begin{equation}
	\begin{aligned}	
	p(d_{\rm{SU},\bar{k}}>&\frac{1}{\sqrt{L^{1-\epsilon}}}) = p(r_{\bar{k},l}>\frac{1}{\sqrt{L^{1-\epsilon}}})^L \\ &= (1-\frac{\mathcal{A}(B(\bc,\frac{1}{\sqrt{L^{1-\epsilon}}})\cap B(\b0,1))}{\mathcal{A}(B(\b0,1))})^L,
	\end{aligned}
	\label{equ:14}
	\end{equation}
	where $B(\ba,b)$ is a circle centered at $\ba$ and with radius $b$, $\b0$ is the center of area of interest, $\bc$ is the location of SU ${\bar{k}}$, $\mathcal{A}(.)$ is the operator that calculates the surface of the input area, and $\cap$ represents the intersection of two areas.
	
	Clearly, the probability is higher when the distance of the SU ${\bar{k}}$ from the center of area of interest increases. Since we want to calculate an upper bound for this probability, we assume that SU ${\bar{k}}$ is on the border of area of interest $B(\b0,1)$. Also, by geometry, it is easy to see that the intersection of the circle of radius $\frac{1}{\sqrt{L^{1-\epsilon}}}$ centered on the border of a circle with a unit radius satisfies:
	\begin{equation}
	\mathcal{A}(B(\bc,\frac{1}{\sqrt{L^{1-\epsilon}}})\cap B(\b0,1)) \geqslant \frac{\mathcal{A}(B(\bc,\frac{1}{\sqrt{L^{1-\epsilon}}}))}{4},
	\label{equ:15}
	\end{equation}
	which is a lower bound for the area of intersection and thus an upper bound for the probability. Using \eqref{equ:15} in \eqref{equ:14}, we have
	\begin{equation}
	p(d_{\rm{SU},\bar{k}}>\frac{1}{\sqrt{L^{1-\epsilon}}}) \leqslant (1-\frac{1}{4{L^{1-\epsilon}}})^L = {\rm{e}}^{\ln {(1-\frac{1}{4{L^{1-\epsilon}}})^L}}.
	\end{equation}
	Let us investigate the limit of the exponent:
	\begin{equation}
	\begin{aligned}	
	\lim_{L\to \infty} \ln {(1-\frac{1}{4{L^{1-\epsilon}}})^L} &= \lim_{L\to \infty} L\ln {(1-\frac{1}{4{L^{1-\epsilon}}})} \\ &= \lim_{L\to \infty} \frac{\ln {(1-\frac{1}{4{L^{1-\epsilon}}})}}{\frac{1}{L}},
	\end{aligned}
	\end{equation}
	which gets the form $\frac{0}{0}$ as $L\to \infty$. Applying L'H\^{o}pital's rule:
	\begin{equation}
	\lim_{L\to \infty} \ln {(1-\frac{1}{4{L^{1-\epsilon}}})^L} = \lim_{L\to \infty} -\frac{(1-\epsilon) L^{\epsilon}}{(1-\frac{1}{4L^{1-\epsilon}})} \to -\infty.
	\end{equation}
	Therefore
	\begin{equation}
	\lim_{L\to \infty} p(d_{\rm{SU},\bar{k}}>\frac{1}{\sqrt{L^{1-\epsilon}}}) \leqslant \lim_{L\to \infty} {\rm{e}}^{\ln {(1-\frac{1}{4{L^{1-\epsilon}}})^L}} = 0,
	\end{equation}
	and consequently $\lim_{L\to \infty} p(d_{\rm{SU},\bar{k}}>\frac{1}{\sqrt{L^{1-\epsilon}}}) = 0$.
\end{proof}

\section{Proof of lemma \ref{lemma4}}
\label{proof lemma 4}
\begin{proof}	
	Similar to the proof of lemma \ref{lemma3}, we have
	\begin{equation}
	\begin{aligned}
	p(d_{\rm{PU},l}>&\frac{1}{\sqrt{\bar{K}^{1+\epsilon}}}) = (1- p(r_{\bar{k},l}<\frac{1}{\sqrt{\bar{K}^{1+\epsilon}}}))^{\bar{K}} \\ &= 
	(1-\frac{\mathcal{A}(B(\bc,\frac{1}{\sqrt{\bar{K}^{1+\epsilon}}})\cap B(\b0,1))}{\mathcal{A}(B(\b0,1))})^{\bar{K}},
	%1-\frac{1}{{\bar{K}^{1+\epsilon}}}.
	\end{aligned}
	\end{equation}
	where definitions of $\mathcal{A}(.)$ and $B(.,.)$ are the same as in Appendix \ref{proof lemma 3}, and $\bc$ is the location of PU $l$.
	Here, we are interested in a lower bound of this probability. Therefore, we need to find an upper bound for the area of intersection between the two circles. It is straightforward to see that assuming PU $l$ at the center leads to a lower bound. Hence,
	\begin{equation}
	p(d_{\rm{PU},l}>\frac{1}{\sqrt{\bar{K}^{1+\epsilon}}}) \geqslant (1-\frac{1}{{\bar{K}^{1+\epsilon}}})^{\bar{K}} = {\rm{e}}^{\ln {(1-\frac{1}{{\bar{K}^{1+\epsilon}}})^{\bar{K}}}}.	
	\end{equation}
	The limit of the exponent is
	
	\begin{equation}
	\begin{aligned}	
	\lim_{{\bar{K}}\to \infty} \ln {(1-\frac{1}{{{\bar{K}}^{1+\epsilon}}})^{\bar{K}}} &= \lim_{{\bar{K}}\to \infty} {\bar{K}}\ln {(1-\frac{1}{{{\bar{K}}^{1+\epsilon}}})} \\ &= \lim_{{\bar{K}}\to \infty} \frac{\ln {(1-\frac{1}{{{\bar{K}}^{1+\epsilon}}})}}{\frac{1}{{\bar{K}}}},
	\end{aligned}
	\end{equation}
	which gets the form $\frac{0}{0}$ as ${\bar{K}}\to \infty$. Applying L'H\^{o}pital's rule:
	\begin{equation}
	\lim_{{\bar{K}}\to \infty} \ln {(1-\frac{1}{{{\bar{K}}^{1+\epsilon}}})^{\bar{K}}} = \lim_{{\bar{K}}\to \infty} -\frac{(1+\epsilon) {\bar{K}}^{-\epsilon}}{1-\frac{1}{{{\bar{K}}^{1+\epsilon}}}} = 0.
	\end{equation}
	Therefore
	\begin{equation}
	\lim_{{\bar{K}}\to \infty} p(d_{\rm{PU},l}>\frac{1}{\sqrt{\bar{K}^{1+\epsilon}}}) \geqslant \lim_{{\bar{K}}\to \infty} {\rm{e}}^{\ln {(1-\frac{1}{{\bar{K}^{1+\epsilon}}})^{\bar{K}}}} = 1,
	\end{equation}	
	which implies $\lim_{{\bar{K}}\to \infty} p(d_{\rm{PU},l}>\frac{1}{\sqrt{\bar{K}^{1+\epsilon}}}) = 1$.
\end{proof}

\section{Proof of theorem \ref{theorem1}}
\label{proof theorem 1}
\begin{proof}
	Without loss of generality, let us consider SU $1$ and assume PU $1$ is the closest PU to SU $1$. Rewriting the interference-temperature constraint of C1 in \eqref{equ:mainopt} for PU $1$, we have
	\begin{equation}
	P_1 \leqslant \frac{I_{th}-\sum_{i=2}^{\bar{K}}P_iF_{1,i}}{F_{1}}.
	\label{equ:theo1proof0}
	\end{equation}
	In the following, we show that the probability that the term $\sum_{i=2}^{\bar{K}}P_iF_{1,i}$ is negligible approaches 1. 
	First, let us show that the joint probability of the events
	$E_1 = \{d_{\mathrm{SU},i}<\frac{1}{\sqrt{L^{1-\epsilon}}}\}$ and $E_2 = \{d_{\mathrm{PU},1}>\frac{1}{\sqrt{\bar{K}^{1+\epsilon}}}\}$ approaches to 1 as $L$ tends to infinity. Note that
	\begin{equation}
	\begin{aligned}	
	p(E_1\cap E_2) &= p(d_{\mathrm{SU},i}<\frac{1}{\sqrt{L^{1-\epsilon}}})p(d_{\mathrm{PU},1}>\frac{1}{\sqrt{\bar{K}^{1+\epsilon}}}).
	\end{aligned}
	\end{equation}	
	From the fact that $K = L^\beta$, with $0<\beta<1$, and that $\bar{K}^2\leq K$, we can write that
	\begin{equation}
	\begin{aligned}	
	p(E_1\cap E_2) &\leq \Big(1 - (1-\frac{1}{L^{1-\epsilon}})^L\Big)\Bigl(1-\frac{1}{L^{\beta\frac{1+\epsilon}{2}}}\Bigr)^{L^{\frac{\beta}{2}}}.\label{eq_theo_proof_3}
	\end{aligned}
	\end{equation}	
	From the derivations of Lemma~\ref{lemma3} and Lemma \ref{lemma4} in previous appendices, we can take the limit of~\eqref{eq_theo_proof_3} and show that
	\begin{equation}
	\begin{aligned}	
	\lim_{L\rightarrow\infty} p(E_1\cap E_2) &=1.\label{eq_theo_proof_3b}
	\end{aligned}
	\end{equation}		
	Let us show now that under the previous condition, the interference term can be made negligible. 
	From Lemma \ref{lemma1} and based on the relation of $F$ and distance, we have
	\begin{equation}
	\sum_{i=2}^{\bar{K}}P_iF_{1,i} \leqslant I_{th}\sum_{i=2}^{\bar{K}}\frac{F_{1,i}}{F_{i}} = I_{th}\sum_{i=2}^{\bar{K}}\frac{d_{\mathrm{SU},i}^{\alpha}}{r_{1,i}^{\alpha}}.
	\label{equ:theo1proof1}
	\end{equation}
	Consider now that $d_{\mathrm{SU},i}<\frac{1}{\sqrt{L^{1-\epsilon}}}$ and $r_{1,i}>\frac{1}{\sqrt{\bar{K}^{1+\epsilon}}}$ (note that $r_{1,i}$ is lower bounded by $d_{\mathrm{PU},1}$). 
	Subsequently, we simplify \eqref{equ:theo1proof1} as
	\begin{equation}
	\sum_{i=2}^{\bar{K}}P_iF_{1,i} \leqslant I_{th}\sum_{i=2}^{\bar{K}}\frac{{{\bar{K}}^{(1+\epsilon){\frac{\alpha}{2}}}}}{{L^{(1-\epsilon)\frac{\alpha}{2}}}}.
	\label{equ:theo1proof2}
	\end{equation}
	Assuming $\alpha=2$, i.e., free space propagation loss, \eqref{equ:theo1proof2} boils down to
	\begin{equation}
	\sum_{i=2}^{\bar{K}}P_iF_{1,i} \leqslant I_{th}\sum_{i=2}^{\bar{K}}\frac{{\bar{K}}^{1+\epsilon}}{L^{1-\epsilon}}=I_{th}(\bar{K}-1)\frac{{\bar{K}}^{1+\epsilon}}{L^{1-\epsilon}}.
	\label{equ:theo1proof3}
	\end{equation}	
	Substituting this result in \eqref{equ:theo1proof0} leads to
	\begin{equation}
	P_1 \leqslant \frac{I_{th}}{F_{1}}(1-\frac{\bar{K}^{2+\epsilon}}{L^{1-\epsilon}}) \approx \frac{I_{th}}{F_{1}}(1-\frac{1}{\lambda}).
	\label{equ:constraintswithlambda}
	\end{equation}	
	Since Lemma~\ref{lemma3} and Lemma \ref{lemma4} are correct for any $0< \epsilon < 1$, we can assume that $0< \epsilon < \frac{2(1-\beta)}{1+\frac{\beta}{2}}$. Thus, for $\lambda \to \infty$, we have
	%$\underset{\lambda \to \infty}{\lim} \sum_{i=2}^{\bar{K}}P_iF_{1,i} \to 0$.	
	\begin{equation}
	P_1 \leqslant \frac{I_{th}}{F_{1}}.
	\end{equation}	
\end{proof}

\bibliographystyle{IEEEtran}
\bibliography{ref.bib}

\begin{thebibliography}{10}
\providecommand{\url}[1]{#1}
\csname url@samestyle\endcsname
\providecommand{\newblock}{\relax}
\providecommand{\bibinfo}[2]{#2}
\providecommand{\BIBentrySTDinterwordspacing}{\spaceskip=0pt\relax}
\providecommand{\BIBentryALTinterwordstretchfactor}{4}
\providecommand{\BIBentryALTinterwordspacing}{\spaceskip=\fontdimen2\font plus
\BIBentryALTinterwordstretchfactor\fontdimen3\font minus
  \fontdimen4\font\relax}
\providecommand{\BIBforeignlanguage}[2]{{%
\expandafter\ifx\csname l@#1\endcsname\relax
\typeout{** WARNING: IEEEtran.bst: No hyphenation pattern has been}%
\typeout{** loaded for the language `#1'. Using the pattern for}%
\typeout{** the default language instead.}%
\else
\language=\csname l@#1\endcsname
\fi
#2}}
\providecommand{\BIBdecl}{\relax}
\BIBdecl

\bibitem{1391031}
S.~{Haykin}, ``Cognitive radio: brain-empowered wireless communications,''
  \emph{IEEE Journal on Selected Areas in Communications}, vol.~23, no.~2, pp.
  201--220, Feb 2005.

\bibitem{788210}
J.~{Mitola} and G.~Q. {Maguire}, ``Cognitive radio: making software radios more
  personal,'' \emph{IEEE Personal Communications}, vol.~6, no.~4, pp. 13--18,
  Aug 1999.

\bibitem{7748543}
A.~{Ali} and W.~{Hamouda}, ``Advances on spectrum sensing for cognitive radio
  networks: Theory and applications,'' \emph{IEEE Communications Surveys
  Tutorials}, vol.~19, no.~2, pp. 1277--1304, Secondquarter 2017.

\bibitem{5963799}
D.~{Bharadia}, G.~{Bansal}, P.~{Kaligineedi}, and V.~K. {Bhargava}, ``Relay and
  power allocation schemes for {OFDM}-based cognitive radio systems,''
  \emph{IEEE Transactions on Wireless Communications}, vol.~10, no.~9, pp.
  2812--2817, Sep. 2011.

\bibitem{8546798}
M.~{Liu}, T.~{Song}, J.~{Hu}, J.~{Yang}, and G.~{Gui}, ``Deep learning-inspired
  message passing algorithm for efficient resource allocation in cognitive
  radio networks,'' \emph{IEEE Transactions on Vehicular Technology}, vol.~68,
  no.~1, pp. 641--653, Jan 2019.

\bibitem{8500158}
M.~{Liu}, T.~{Song}, and G.~{Gui}, ``Deep cognitive perspective: Resource
  allocation for noma-based heterogeneous iot with imperfect sic,'' \emph{IEEE
  Internet of Things Journal}, vol.~6, no.~2, pp. 2885--2894, April 2019.

\bibitem{7336495}
E.~{Lagunas}, S.~K. {Sharma}, S.~{Maleki}, S.~{Chatzinotas}, and
  B.~{Ottersten}, ``Resource allocation for cognitive satellite communications
  with incumbent terrestrial networks,'' \emph{IEEE Transactions on Cognitive
  Communications and Networking}, vol.~1, no.~3, pp. 305--317, Sep. 2015.

\bibitem{7510839}
E.~{Lagunas}, S.~{Maleki}, S.~{Chatzinotas}, M.~{Soltanalian}, A.~I.
  {Pérez-Neira}, and B.~{Oftersten}, ``Power and rate allocation in cognitive
  satellite uplink networks,'' in \emph{2016 IEEE International Conference on
  Communications (ICC)}, May 2016, pp. 1--6.

\bibitem{5955145}
G.~{Alnwaimi}, K.~{Arshad}, and K.~{Moessner}, ``Dynamic spectrum allocation
  algorithm with interference management in co-existing networks,'' \emph{IEEE
  Communications Letters}, vol.~15, no.~9, pp. 932--934, Sep. 2011.

\bibitem{7516556}
A.~K. {Gupta}, J.~G. {Andrews}, and R.~W. {Heath}, ``On the feasibility of
  sharing spectrum licenses in mmwave cellular systems,'' \emph{IEEE
  Transactions on Communications}, vol.~64, no.~9, pp. 3981--3995, Sep. 2016.

\bibitem{7039232}
P.~{Luoto}, P.~{Pirinen}, M.~{Bennis}, S.~{Samarakoon}, S.~{Scott}, and
  M.~{Latva-aho}, ``Co-primary multi-operator resource sharing for small cell
  networks,'' \emph{IEEE Transactions on Wireless Communications}, vol.~14,
  no.~6, pp. 3120--3130, June 2015.

\bibitem{7060478}
S.~{Maleki}, S.~{Chatzinotas}, B.~{Evans}, K.~{Liolis}, J.~{Grotz},
  A.~{Vanelli-Coralli}, and N.~{Chuberre}, ``Cognitive spectrum utilization in
  ka band multibeam satellite communications,'' \emph{IEEE Communications
  Magazine}, vol.~53, no.~3, pp. 24--29, March 2015.

\bibitem{maral2011satellite}
G.~Maral and M.~Bousquet, \emph{Satellite communications systems: systems,
  techniques and technology}.\hskip 1em plus 0.5em minus 0.4em\relax John Wiley
  \& Sons, 2011.

\bibitem{louchart2019resource}
A.~{Louchart}, P.~{Ciblat}, and P.~d.~{Kerret}, ``Resource optimization for
  cognitive satellite systems with incumbent terrestrial receivers,'' in
  \emph{27th European Signal Processing Conference (EUSIPCO)}, Sep. 2019, pp.
  1--5.

\bibitem{8537943}
E.~{Tohidi}, M.~{Coutino}, S.~P. {Chepuri}, H.~{Behroozi}, M.~M. {Nayebi}, and
  G.~{Leus}, ``Sparse antenna and pulse placement for colocated {MIMO} radar,''
  \emph{IEEE Transactions on Signal Processing}, vol.~67, no.~3, pp. 579--593,
  Feb 2019.

\end{thebibliography}

\end{document}